\documentclass[11pt]{article}%
\usepackage{amsfonts}
\usepackage{amsmath}
\usepackage{amssymb}
\usepackage{graphicx}%
\providecommand{\U}[1]{\protect\rule{.1in}{.1in}}
\newtheorem{theorem}{Theorem}
\newtheorem{acknowledgement}[theorem]{Acknowledgement}

\newtheorem{corollary}[theorem]{Corollary}

\newenvironment{proof}[1][Proof]{\noindent\textbf{#1.} }{\ \rule{0.5em}{0.5em}}
\begin{document}

\title{A Symplectic Interpretation of the Separability of Gaussian Mixed States}
\author{Maurice A. de Gosson\thanks{maurice.de.gosson@univie.ac.at}\\University of Vienna\\Faculty of Mathematics (NuHAG)\\Vienna, AUSTRIA}
\maketitle

\begin{abstract}
Using symplectic methods and the Wigner formalism we prove a refinement of a
criterion due to Werner and Wolf for the separability of bipartite Gaussian
mixed states in an arbitrary number of dimensions. We use our result to show
that one can characterize separability by comparing these states with
separable pure Gaussian states.

\end{abstract}

Let $\widehat{\rho}$ be the density matrix of a Gaussian bipartite mixed
quantum state on Hilbert space $L^{2}(\mathbb{R}^{n})=L^{2}(\mathbb{R}^{n_{A}%
})\otimes L^{2}(\mathbb{R}^{n_{B}})$; we write $z=z_{A}\oplus z_{B}$ where
$z_{A}=(x_{A},p_{A})$ and $z_{B}=(x_{B},p_{B})$. We write $J_{AB}=J_{A}\oplus
J_{B}$ where $J_{A}$ (\textit{resp.} $J_{B}$) is the standard symplectic
matrix on $\mathbb{R}^{2n_{A}}$ (\textit{resp}. $\mathbb{R}^{2n_{B}}$). Let
\begin{equation}
\rho(z)=\left(  \tfrac{1}{2\pi}\right)  ^{n}(\det\Sigma)^{-1/2}e^{-\frac{1}%
{2}\Sigma^{-1}z\cdot z}\label{Gauss1}%
\end{equation}
be the Wigner distribution of $\widehat{\rho}$; the covariance matrix $\Sigma$
satisfies the quantum condition
\begin{equation}
\Sigma+\frac{i\hbar}{2}J_{AB}\geq0.\label{quantum1}%
\end{equation}
Elaborating on previous work \cite{HHH1,HHH2,Peres,Simon} Werner and Wolf
\cite{ww1} have shown that a Gaussian state with Wigner distribution
(\ref{Gauss1}) is separable if and only there exist positive symmetric
matrices $\Sigma_{A}$ and $\Sigma_{B}$ such that
\begin{equation}
\Sigma\geq\Sigma_{.A}\oplus\Sigma_{B}\text{ \ , }\Sigma_{A}+\frac{i\hbar}%
{2}J_{A}\geq0\text{ \ , }\Sigma_{B}+\frac{i\hbar}{2}J_{B}\geq0.\label{ww}%
\end{equation}
Lami \textit{et al.} \cite{lami} have recently shown that (\ref{ww}) is
equivalent to the condition%
\begin{equation}
\Sigma\geq\Sigma_{A}\oplus\frac{i\hbar}{2}J_{B}.\label{lami}%
\end{equation}

We have the following refinement of (\ref{ww}):

\begin{theorem}
\label{Prop1}The Gaussian bipartite state $\widehat{\rho}$ is separable if and
only if there exist positive definite symplectic matrices $P_{A}$ and $P_{B}$
such that
\begin{equation}
\Sigma\geq\frac{\hbar}{2}\left[  P_{A}\oplus P_{B}\right]  .\label{MC}%
\end{equation}

\end{theorem}

\begin{proof}
Define $S_{A}\in\operatorname*{Sp}(n_{A})$ and $S_{B}\in\operatorname*{Sp}%
(n_{B})$ by $P_{A}=(S_{A}^{T}S_{A})^{-1}$ and $P_{B}=(S_{A}^{T}S_{B})^{-1}$.
Let us show that the condition
\[
\Sigma\geq\frac{\hbar}{2}\left[  (S_{A}^{T}S_{A})^{-1}\oplus(S_{B}^{T}%
S_{B})^{-1}\right]
\]
is sufficient. The matrices $\Sigma_{A}=\frac{\hbar}{2}(S_{A}^{T}S_{A})^{-1}%
$and $\Sigma_{B}=\frac{\hbar}{2}(S_{B}^{T}S_{B})^{-1}$ are quantum covariance
matrices: for instance the condition $\frac{\hbar}{2}(S_{A}^{T}S_{A}%
)^{-1}+\frac{i\hbar}{2}J_{A}\geq0$ is equivalent to $I_{A}+iJ_{A}\geq0$ which
is trivially satisfied. It follows from (\ref{ww}) that $\widehat{\rho}$ is
separable. The inequality (\ref{MC}) is also necessary: the condition
$\Sigma_{A}+\frac{i\hbar}{2}J_{A}\geq0$ in (\ref{ww}) implies that the
symplectic capacity of the the covariance ellipsoid $\Omega_{A}=\{z_{A}%
:\frac{1}{2}\Sigma_{A}z_{A}\cdot z_{A}\leq1\}$ is at least $\pi\hbar$; it
follows \cite{Birk,FOOP,physreps} that $\Omega_{A}$ contains a quantum blob
\cite{blobs}, that is, there exists $S_{A}\in\operatorname*{Sp}(n_{A})$ such
that $S_{A}^{-1}B^{2n_{A}}(\sqrt{\hbar})\subset\Omega_{A}$ where $B^{2n_{A}%
}(\sqrt{\hbar})$ is the centered ball in $\mathbb{R}^{2n_{A}}$ with radius
$\sqrt{\hbar}$. We must thus have $\Sigma_{A}\geq\frac{\hbar}{2}(S_{A}%
^{T}S_{A})^{-1}$. Similarly there exists $S_{B}\in\operatorname*{Sp}(n_{A})$
such that $\Sigma_{B}\geq\frac{\hbar}{2}(S_{B}^{T}S_{B})^{-1}$ and hence the
inequality (\ref{MC}) must hold. 
\end{proof}

The result above is an improvement of the Werner--Wolf result: in order to
determine whether a Gaussian state is separable the latter require that one
determines real symmetric matrices $\Sigma_{A}$ and $\Sigma_{B}$ satisfying
the inequalities (\ref{ww}) while we only require that one finds two positive
definite symplectic matrices $P_{A}$ and $P_{B}$ satisfying (\ref{MC}).  

Suppose we have equality in (\ref{MC}). Then the state $\widehat{\rho}$ is a
pure state, in fact the tensor product $\widehat{S}_{A}^{-1}\phi_{A}%
\otimes\widehat{S}_{B}^{-1}\phi_{B}$ where
\begin{align*}
\phi_{A}(x_{A})  &  =(\pi\hbar)^{-n_{A}/4}e^{-|x_{A}|^{2}/2\hbar}\\
\phi_{B}(x_{B})  &  =(\pi\hbar)^{-n_{B}/4}e^{-|x_{B}|^{2}/2\hbar}%
\end{align*}
are the standard coherent states on $\mathbb{R}^{n_{A}}$ and $\mathbb{R}%
^{n_{B}}$, and $\widehat{S}_{A}\in\operatorname*{Mp}(n_{A})$ (\textit{resp}.
$\widehat{S}_{B}\in\operatorname*{Mp}(n_{B})$) is anyone of the two
metaplectic operators covering\ $S_{A}$ (\textit{resp.} $S_{B}$). In fact, the
Wigner distribution (\ref{Gauss1}) becomes in this case
\begin{align*}
\rho(z)  &  =(\pi\hbar)^{-n}e^{-\frac{1}{\hbar}(S_{A}^{T}S_{A}z_{A}\cdot
z_{A}+S_{B}^{T}S_{B}z_{B}\cdot z_{B})}\\
&  =W_{A}\phi_{A}(S_{A}z_{A})W_{B}\phi_{B}(S_{B}z_{B})
\end{align*}
where $W_{A}\phi_{A}$ is the Wigner transform of and $W_{B}\phi_{B}$ that of
$\phi_{B}$ (similarly defined). It follows from the symplectic covariance
property \cite{Wigner} of the Wigner transform that
\[
W_{A}\phi_{A}(S_{A}\cdot)=W_{A}(\widehat{S}_{A}^{-1}\phi_{A})\text{ \ ,
\ }W_{B}\phi_{B}(S_{B}\cdot)=W_{A}(\widehat{S}_{B}^{-1}\phi_{B})
\]
hence $\rho(z)$ is the Wigner transform of $\widehat{S}_{A}^{-1}\phi
_{A}\otimes\widehat{S}_{B}^{-1}\phi_{B}$. The converse of this property is
trivial. Notice that the states $\widehat{S}_{A}^{-1}\phi_{A}$ and
$\widehat{S}_{B}^{-1}\phi_{B}$ are easily calculated \cite{Birk,Wigner}: they
are explicitly given by
\begin{align*}
\widehat{S}_{A}^{-1}\phi_{A}(x_{A})  &  =(\pi\hbar)^{-n_{A}/4}(\det
X_{A})^{1/4}e^{-\frac{1}{2\hbar}(X_{A}+iY_{A})x_{A}\cdot x_{A}}\\
\widehat{S}_{B}^{-1}\phi_{B}(x_{B})  &  =(\pi\hbar)^{-n_{B}/4}(\det
X_{B})^{1/4}e^{-\frac{1}{2\hbar}(X_{B}+iY_{B})x_{B}\cdot x_{B}}%
\end{align*}
where the real symmetric matrices $X_{A}>0$, $X_{B}>0$ and $Y_{A},Y_{B}$ are
obtained by solving the identities%
\begin{align*}
S_{A}^{T}S_{A}  &  =%
\begin{pmatrix}
X_{A}+Y_{A}X_{A}^{-1}Y_{A} & Y_{A}X_{A}^{-1}\\
X_{A}^{-1}Y_{A} & X_{A}^{-1}%
\end{pmatrix}
\\
S_{B}^{T}S_{B}  &  =%
\begin{pmatrix}
X_{B}+Y_{B}X_{B}^{-1}Y_{B} & Y_{B}X_{B}^{-1}\\
X_{B}^{-1}Y_{B} & X_{B}^{-1}%
\end{pmatrix}
.
\end{align*}

More generally Proposition \ref{Prop1} implies that $\widehat{\rho}$ is
separable if and only if it dominates a tensor product of two pure Gaussian states:

\begin{corollary}
The Gaussian bipartite state $\widehat{\rho}$ is separable if and only if
there exist $S_{A}\in\operatorname*{Sp}(n_{A})$ and $S_{B}\in
\operatorname*{Sp}(n_{B})$ such that $S_{AB}=S_{A}\oplus S_{B}$ satisfies
\begin{equation}
\rho(S_{AB}^{-1}z)\geq\mu(\widehat{\rho})W_{A}\phi_{A}(z_{A})W_{B}\phi
_{B}(z_{B}). \label{corab2}%
\end{equation}
where $\phi_{A}$ (resp. $\phi_{B}$) is the standard Gaussian and
\[
\mu(\widehat{\rho})=\left(  \frac{\hbar}{2}\right)  ^{n}(\det\Sigma)^{-1/2}%
\]
is the purity of $\widehat{\rho}$.
\end{corollary}

\begin{proof}
In view of Proposition \ref{Prop1} $\widehat{\rho}$ is separable if and only
if condition (\ref{MC})
\begin{equation}
\Sigma\geq\frac{\hbar}{2}\left[  (S_{A}^{T}S_{A})^{-1}\oplus(S_{B}^{T}%
S_{B})^{-1}\right]
\end{equation}
holds for some $S_{A}\in\operatorname*{Sp}(n_{A})$ and $S_{B}\in
\operatorname*{Sp}(n_{B})$. Suppose it is the case; by definition
(\ref{Gauss1}) of $\rho$ we then have
\[
\rho(z)\geq(2\pi)^{-n}(\det\Sigma)^{-1/2\infty}e^{-\frac{1}{\hbar}S_{A}%
^{T}S_{A}z_{A}\cdot z_{A}}e^{-\frac{1}{\hbar}S_{B}^{T}S_{B}z_{B}\cdot z_{B}%
)}.
\]
We have \cite{Birk,Wigner}%
\begin{align*}
W_{A}\phi_{A}(S_{A}^{-1}z_{A})  &  =(\pi\hbar)^{-n_{A}}e^{-\frac{1}{\hbar
}|z_{A}|^{2}}\\
W_{B}\phi_{B}(S_{B}^{-1}z_{B})  &  =(\pi\hbar)^{-n_{B}}e^{-\frac{1}{\hbar
}|z_{B}|^{2}}%
\end{align*}
and hence%
\begin{equation}
\rho(z)\geq\left(  \frac{\hbar}{2}\right)  ^{n}(\det\Sigma)^{-1/2}W_{A}%
\phi_{A}(S_{A}^{-1}z_{A})W_{B}\phi_{B}(S_{B}^{-1}z_{B}) \label{corab3}%
\end{equation}
which shows that (\ref{corab2}) must hold if the state $\widehat{\rho}$ is
separable. Suppose conversely that this inequality holds. Then we must have
\[
e^{-\frac{1}{2}\Sigma^{-1}z\cdot z}\leq e^{-\frac{1}{\hbar}S_{A}^{T}S_{A}%
z_{A}\cdot z_{A}}e^{-\frac{1}{\hbar}S_{B}^{T}S_{B}z_{B}\cdot z_{B})}%
\]
which is equivalent to our condition (\ref{MC}).
\end{proof}

\begin{acknowledgement}
This work has been financed by the Austrian Research Foundation FWF (Grant
number P27773).
\end{acknowledgement}

\begin{acknowledgement}
I wish to express my gratitude to Nuno Dias and J. Prata (Lisbon) for having
pointed out an erroneous statement in the Corollary to Theorem \ref{Prop1}. 
\end{acknowledgement}

\end{document}